\documentclass{amsart}

\usepackage{graphicx}
\usepackage{natbib}
\usepackage{amssymb, amsmath}

\newcommand{\EE}{\mathbb{E}}
\newcommand{\PP}{\mathbb{P}}
\newcommand{\RR}{\mathbb{R}}

\newcommand{\Qit}{Q_t^i}
\newcommand{\Pit}{P_t^i}
\newcommand{\Git}{G_t^i}

\newcommand{\calG}{\mathcal{G}}
\newcommand{\frakG}{\mathfrak{G}}
\newcommand{\calQ}{\mathcal{Q}}

\newcommand{\cond}{\, | \,}
\newcommand{\poi}[1]{\mathrm{Poisson}(#1)}

\newcommand{\upth}{^{\textrm{th}}}

\newcommand{\tp}{TP(2)}
\newcommand{\Bin}{\operatorname{Binomial}}
\newcommand{\var}{\operatorname{Var}}
\newcommand{\tmrca}{T}
\newcommand{\nmrca}{M}

\newcommand{\sorted}{F}
\newcommand{\arxiv}[1]{#1}
\newcommand{\noarxiv}[1]{}
\newcommand{\QEDclosed}{\mbox{\rule[0pt]{1.3ex}{1.3ex}}}

\noarxiv{
\newenvironment{proof}{\noindent \emph{Proof:} \newline \noindent}{\QEDclosed}
}

\noarxiv{
\journal{Theoretical Population Biology}
}

\newtheorem{pro}{Proposition}
\arxiv{
\newtheorem{defn}{Definition}

\newtheorem{lem}{Lemma}
\newtheorem{cor}{Corollary}
}

\arxiv{
\title[Genetics and the MRCA]{To what extent does genealogical ancestry
imply genetic ancestry?}

 \author{Frederick A. \ Matsen}
 \address{
  Department of Statistics \#3860 \\
  University of California at Berkeley \\
  367 Evans Hall \\
  Berkeley, CA 94720-3860 \\
  U.S.A.}
\email{matsen@berkeley.edu}
\urladdr{http://www.stat.berkeley.edu/users/matsen/}

\author{Steven N.\ Evans}
\address{
  Department of Statistics \#3860 \\
  University of California at Berkeley \\
  367 Evans Hall \\
  Berkeley, CA 94720-3860 \\
  U.S.A.}
\email{evans@stat.Berkeley.EDU}
\urladdr{http://www.stat.berkeley.edu/users/evans/}

\thanks{FAM is funded by the Miller Institute for Basic Research in
Science at the University of California, Berkeley.  SNE is supported
in part by NSF grant DMS-0405778. Part of this research was conducted
during a visit to the Pacific Institute for the Mathematical
Sciences.}

\noarxiv{
} 
}

\begin{document}
\noarxiv{
\begin{frontmatter}

\title{To what extent does genealogical ancestry imply genetic ancestry?}

 \author[Berk]{Frederick A. \ Matsen\corauthref{cor}}
 \corauth[cor]{Corresponding author.}
 \ead{matsen@berkeley.edu}
 \ead[url]{http://www.stat.berkeley.edu/users/matsen/}
 \author[Berk]{Steven N.\ Evans}
 \ead{evans@stat.Berkeley.EDU}
 \ead[url]{http://www.stat.berkeley.edu/users/evans/}

 \address[Berk]{
  Department of Statistics, University of California at Berkeley, 367
  Evans Hall, Berkeley, CA 94720-3860 U.S.A.}

\end{frontmatter}
}

\arxiv{ \maketitle }

\noarxiv{ \newpage }

\arxiv { \begin{abstract} }
  \noarxiv { \noindent \textbf{ Abstract } }

  Recent statistical and computational analyses have shown that a
  genealogical most recent common ancestor (MRCA) may have lived in
  the recent past \citep{chang99, rohde04}. However, coalescent-based
  approaches show that genetic most recent common ancestors 
  for a given non-recombining locus are typically much more
  ancient \citep{kingman82a, kingman82b}. It is not
  immediately clear how these two perspectives interact. This paper
  investigates relationships between the number of descendant
  alleles of an ancestor allele and the number of genealogical descendants
  of the individual who possessed that allele for a simple
  diploid genetic model extending the genealogical model of
  \citet{chang99}.
\arxiv{ \end{abstract} }

\noarxiv{ \newpage }
 
\section{Introduction and model}
Joseph Chang's 1999 paper \citep{chang99} showed that a well-mixed
closed diploid population of $n$ individuals will have a genealogical common
ancestor in the recent past. Specifically, the paper showed that if
$\tmrca_n$ is the number of generations back to the most recent common
ancestor (MRCA) of the population, then $\tmrca_n$ divided by $\log_2
n$ converges to one in probability as $n$ goes to infinity. His paper initiated a discussion
in which many of the leading figures of population genetics expressed
interest in the relationship between the genealogical and genetic
perspectives for such models \citep{discussion}. For example, Peter Donnelly wrote
``[r]esults on the extent to which common ancestors, in the sense of
[Chang's] paper, are ancestors in the genetic sense... would also be
of great interest'' \citep{discussion}. Every other discussant also
either discussed the relationship of Chang's work to genetics or
expressed interest in doing so.

Given this interest, surprisingly little work has been done
specifically about the interplay between the two perspectives. Wiuf
and Hein, in their reply, wrote three paragraphs containing some
simple initial observations \citep{discussion}. Some simulation work
has been done by \citet{socsim} with a more realistic population model.  
In a related though different vein, \citet{mohlesagitov03} derived limiting results for the diploid
coalescent, in the classical setting of a small sample from a
large population.

In an interesting series of papers, Derrida, Manrubia, Zanette, and
collaborators \citep{derridaEAStatProp99, derridaEAPhysica00, 
derridaEATPB00, manrubiaEAAmSci03} have investigated the
distribution of the number of repetitions of ancestors in a
genealogical tree, as well as the degree of concordance between the
genealogical trees for two distinct individuals. Our paper, on the
other hand, is concerned with correlations between the number of
genealogical descendants of an individual and the number of descendant
alleles of that individual. The interesting time-frame in our paper
is different than theirs: they focus on the period substantially after
$\tmrca_n$, while for us any interesting correlation is erased with
high probability after time about $1.77 \tmrca_n$.


Our paper attempts to  connect the 
genealogical and genetic points of view by investigating several
different questions concerning the interaction of  genealogical
ancestry and genetic ancestry in a diploid model incorporating Chang's
model. In classical Wright-Fisher fashion, we consider $2n$
alleles contained in $n$ diploid individuals.  Each discrete generation
forward in time, every individual selects two alleles from the
previous generation independently and uniformly to ``inherit.'' If an
individual $X$ at time $t$ inherits genetic information from an
individual $Y$
at time $t-1$, then we consider $Y$ to be a
``parent'' of $X$ in the genealogical sense.  As with Chang's model,
the two parents are permitted to be the same individual
and each allele of a child may descend
from the same parent allele.  We illustrate the basic operation of the
model in Figure~\ref{fig:example}. Each individual is represented as a
circle, and each of a given individual's alleles are represented as dots
within the circle. Time increases down the figure and inheritance of
alleles is represented by lines connecting them. 

\begin{figure}[h]
\label{fig:example}
\arxiv{
\includegraphics[angle=0,scale=.45]{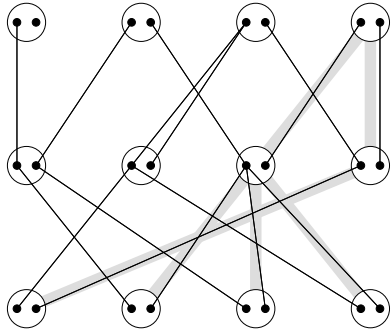}
}
\caption{An example instance of our model with four individuals and
three generations. Time increases moving down the diagram. 
  The two alleles of each individual are depicted as two dots
  within the larger circles; a thin black line indicates
  genetic inheritance, i.e. the lower allele is descended from the
  upper allele. 
  This sample genealogy demonstrates that the genealogical
  MRCA need not have
  any genetic relation to present-day individuals. 
  The individual at the far right on the top row is in
  this case the (unique) MRCA as demonstrated by the thick gray lines,
  however none of its genetic material is passed onto the present day.
  } 
\end{figure}

We have chosen notation in order to fit with Chang's original article.
The initial generation will be denoted $t=0$ and other
generations will be counted forwards in time; thus the parents of the
$t=1$ generation will be in the $t=0$ generation, and so on. The $n$
individuals of generation $t$ will be denoted $I_{t,1}, \ldots,
I_{t,n}$. The two alleles present at a given locus of
 individual $I_{t,i}$ will be labeled
$A_{t,i,1}$ and $A_{t,i,2}$. Using this notation, each allele
$A_{t,i,c}$ of generation $t$ selects an allele $A_{t-1,j,d}$ uniformly and
independently from all of the alleles of the previous generation;
given such a choice we say that allele $A_{t,i,c}$ 
is \emph{descended genetically} from allele
$A_{t-1,j,d}$. 
We define more distant ancestry recursively: allele $A_{t,i,c}$ is descended
from allele $A_{t',j,d}$ if $t > t'$ and there exists a $k$ and $e$ such that
allele $A_{t,i,c}$ is descended genetically from allele
$A_{t-1,k,e}$ and allele $A_{t-1,k,e}$ is
descended from or is the same as allele $A_{t',j,d}$.

One can make a similar recursive definition of genealogical ancestry
that matches Chang's notion of ancestry: individual $I_{t,i}$ is
\emph{descended genealogically} from individual $I_{t',j}$ if $t > t'$
and there exists a $k$ such that individual $I_{t,i}$ is a parent of
individual $I_{t-1,k}$ and individual $I_{t-1,k}$ is descended from or
is the same as individual $I_{t',j}$.  

Define $\calQ_t^{i}$ to be the alleles that are genetic descendants at
time $t$ of the two alleles present in individual $I_{0,i}$,
and let $Q_t^{i}$ be the number
of such alleles. We will call the elements of $\calQ_t^{i}$ the
\emph{descendant alleles} of individual $I_{0,i}$.
Define $\calG_t^i$ to be the genealogical descendants at time $t$ of the
individual $I_{0,i}$, and let $G_t^i$ be the number
of such individuals. We will say that a 
(genealogical) most recent common ancestor (MRCA) first
appears at time $t$ if there is an
individual $I_{0,i}$ in the population
at time $0$ such that  $\Git = n$ and $G_s^j
< n$ for all $j$ and $s < t$; that is, individual $i$ 
in generation $0$ is a genealogical
ancestor of all individuals in generation $t$, but there
is no individual in generation $0$ that is a
a genealogical ancestor of all individuals in any 
generation previous to generation $t$.
Let $\tmrca_n$ denote the generation number
at which the MRCA first appears. The main conclusion of Chang's 1999 paper is that the
ratio $\tmrca_n / \log_2 n$ converges to one in probability as $n$
tends to infinity.

Our intent is to investigate the degree to which genealogical
ancestry implies genetic ancestry. Unsurprisingly, historical
individuals with more genealogical descendants will have more
descendant alleles in expectation: in
Proposition~\ref{pro:monot} we show that $\EE[Q_t^i \cond G_t^i = k]$
is a super-linearly increasing function in $k$. 
However, in any realization of the stochastic process,
individuals with more genealogical descendants need not have more
descendant alleles. For example, in Figure~\ref{fig:example} we show a
case where the MRCA has no genetic relationship to any present day
individuals. In the above
notation, $G_2^4 = n = 4$ and yet $Q_2^4 = 0$. 


Another approach is based on the rank of $\Git$. Loosely speaking, we
are interested in the number of descendant alleles of the
generation-$t$ individual with the $x$th most genealogical
descendants.  More rigorously, we consider the renumbering
(opposite to the way rank is typically defined in statistics)
$\sorted(t,1), \ldots, \sorted(t,n)$ of the indices $1, \ldots, n$
such that
\[
G_t^{\sorted(t,1)} \geq \cdots \geq G_t^{\sorted(t,n)}
\]
and if $G_t^{\sorted(t,i)} = G_t^{\sorted(t,j)}$ then fix
$\sorted(t,i) < \sorted(t,j)$ when $i < j$. 
We then investigate $\left\{Q^{\sorted(t,k)}: 1 \leq k \leq n
\right\}$.
These quantities give us concrete information about our main question
in a relative sense: how much do individuals with many genealogical
descendants contribute to the genetic makeup of present-day
individuals compared to those with only a few? In
Figure~\ref{fig:qOverTime} we simulate our process 10000
times and then take an average for each time step, approximating 
$\EE\left[Q_t^{\sorted(t,k)}\right]$.  

\begin{figure}[h]
  \begin{center}
    \arxiv{
    \includegraphics[width=10cm]{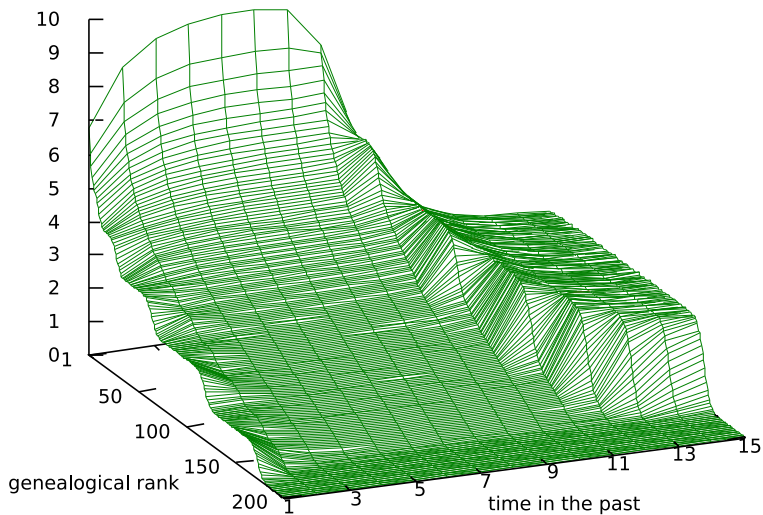}
    }
  \end{center}
  \caption{The expected number of descendant alleles from 
  historical individuals
  sorted by number of genealogical descendants.
  Results by simulation of a population of size 200.
  For example, the value at
  ``genealogical rank'' 50 and time 2 is the expected number of
  alleles in the current population which descend
  from one or other of the two
  alleles present in the 
  individual two generations ago who had no more
  genealogical descendants in the present population
  than did 49 other individuals in the population two generations ago.
  As described in the text, this curve attains an
  interesting characteristic shape around generation 3 that lasts
  until generation 8. We investigate that shape in
  Figure~\ref{fig:branching} and Proposition~\ref{pro:main}.}
  \label{fig:qOverTime}
\end{figure}

After several generations, the curve depicting
$\EE\left[Q_t^{\sorted(t,k)}\right]$
acquires a characteristic shape which persists for
some time, in this figure between time 3 and time 8.
In order to explain what this curve is, we need to
introduce some elementary facts about branching processes.

Recall that a branching process is a 
discrete time Markov process that tracks the population
size of an idealized population \citep{athreyaNeyBranching72, MR2059709}.  
Each individual of generation $t$
produces an independent random number of offspring in generation $t+1$
according to some fixed probability distribution (the {\em offspring
distribution}). This distribution is the same across all
individuals. We will use the Poisson(2) branching process where the
offspring distribution is Poisson with mean $2$ and write $B_t$
for the number of individuals in the
$t\upth$ generation
starting with one individual at time $t=0$.  
It is a standard fact that
the random variables $W_t = B_t / 2^t$ converges almost surely as
$t \rightarrow \infty$ to a
random variable $W$ that is
strictly positive on the
event that the branching process doesn't die out
(that is, on the event that $B_t$ is strictly positive for all $t \ge 0$)
-- cf. Theorem~8.1 of \citep{athreyaNeyBranching72}.  Denote by
$R$ the distribution of the limit random variable $W$. The probability
measure $R$ is diffuse except for an atom at $0$ (that is,
$0$ is the only point to which $R$ assigns non-zero mass).  Also,
the support of $R$ is the whole of $\RR_+$ (that is, every open 
sub-interval of $\RR_+$ is assigned strictly positive mass by $R$).

Returning to our discussion of
$\EE\left[Q_t^{\sorted(t,k)}\right]$, define a 
non-increasing function
$\gamma_{t,n}(c): (0,1) \rightarrow \RR_+$ by 
\[
\gamma_{t,n}(c) 
= \mathbb{E}\left[Q_t^{F(t,\lfloor cn \rfloor})\right],
\]
and define a non-increasing, continuous function 
$\beta: (0,1) \rightarrow \RR_+$ by
\begin{equation}
\label{E:def_beta}
\begin{split}
\beta(c) 
& = \min\{r \ge 0: R\left( (r/2,\infty) \right) \le c\} \\
& = \min\{ r \ge 0: R\left([0,r/2]\right) \ge 1-c\}. \\
\end{split}
\end{equation}
That is, $\beta(c)$ is the $(1-c)^{\mathrm{th}}$ quantile
of $2W$, where the random variable
$W$ is the limit of the normalized
Poisson(2) branching process introduced above. Note
that the function $\beta$ is strictly
decreasing on the interval $(0, 1 - R(\{0\}))$;
that is, $\beta(c)$ is the {\bf unique} value $r$
for which $R( (r/2,\infty)) = c$ when $0 < c < 1 - R(\{0\})$. 
We see experimentally
that $\gamma_{6,200}$ is quite close to $\beta$ in
Figure~\ref{fig:branching}, and establish a convergence result
in Proposition~\ref{pro:main}.
Although a closed-form
expression for the distribution $R$ is not available,
there is a considerable amount known about this classical object
\citep{vanMieghemLimitRV05}. Note that the long-time behavior in
Figure~\ref{fig:qOverTime} is easily explained: it is simply the
uniform distribution across only the common ancestors, that form $1 -
e^{-2} \simeq 0.864$ of the population.

\begin{figure}
  \begin{center}
    \arxiv{
    \includegraphics[height=5cm]{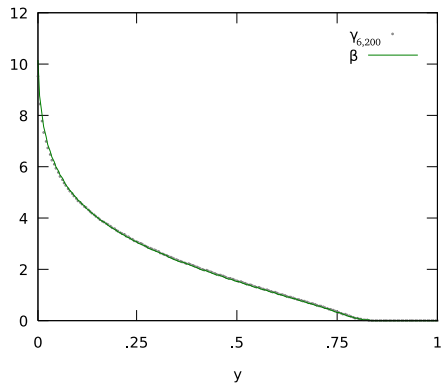}
    }
  \end{center}
  \caption{A plot of $\gamma_{6,200}$ and $\beta$, showing
  experimentally that the characteristic shape in
  Figure~\ref{fig:qOverTime} is very close to the ``tail-quantile'' curve
  of a normalized Poisson(2) branching process. 
  The curve for $\gamma_{6,200}$ was taken from Figure~\ref{fig:qOverTime}.
  To construct the curve for
  $\beta$, we wrote a subroutine that simulated 200 Poisson(2) branching
  processes simultaneously, then sorted the normalized results after 10
  generations. This
  subroutine was run 10000 times and the average was taken. Note that
  the distribution had stabilized after 10 generations.
  }
  \label{fig:branching}
\end{figure}

Thus far we have examined the connection between
genealogical ancestry and genetic ancestry in the population as a
whole; one
may wonder about the number of descendants of the MRCA itself.
Unfortunately, the story there is
not as simple as could be desired. For example, there are usually
multiple MRCAs appearing (by definition) in the same generation, and the
expected number depends on $n$ in a surprising way
(see Figure~\ref{fig:totNMrca}). We investigate this genealogical
issue and related genetic questions
in Section~\ref{sec:mrca}.

\section{Monotonicity of the number of descendant alleles in terms of genealogy}

In this section we prove the following result.
\begin{pro}
  \label{pro:monot}
  For each time $t \ge 2$, the function
  $k \mapsto k^{-1} \, \EE[Q_t^i \cond G_t^i = k]$,  $0 \leq k \leq n$,
  is strictly increasing.
\end{pro}

The key observation in the proof of Proposition \ref{pro:monot}
will be that the random variables $G_t^i$ and $G_{t+1}^i$ enjoy the
property of {\em total positivity} investigated extensively in the
statistical literature following \citet{karlinTotalPositivity68} (see,
for example, \citep{MR650893}). 

\begin{defn}
A pair of random variables $(X,Y)$ has a \emph{strict
\tp}\ joint distribution if
\[
\PP\left\{ X=x,Y=y \right\} 
\PP\left\{ X=x',Y=y' \right\} >
\PP\left\{ X=x,Y=y' \right\} 
\PP\left\{ X=x',Y=y \right\} 
\]
for all $x < x'$ and $y < y'$ such that the left-hand side is strictly
positive.
\end{defn}

The proof of the next result is clear.

\begin{lem}
The following are equivalent to strict \tp\ 
for $x < x'$ and $y < y'$:
\begin{align}
  \frac{\PP\left\{ Y=y' \cond X=x' \right\}} {\PP\left\{ Y=y' \cond
  X=x \right\}} & >
  \frac{\PP\left\{ Y=y \cond X=x' \right\}} {\PP\left\{ Y=y \cond X=x \right\}}
  \label{eq:tp2c} \\
  \frac{\PP\left\{ X=x' | Y=y' \right\}} {\PP\left\{ X=x' \cond Y=y
  \right\}} & >
  \frac{\PP\left\{ X=x | Y=y' \right\}} {\PP\left\{ X=x \cond Y=y
  \right\}}.
  \label{eq:tp2d}
\end{align}
\end{lem}

\begin{lem}
  \label{lem:tp2}
  The pair $(G_t^i, G_{t+1}^i)$ has a strict \tp\ joint distribution. 
\end{lem}
\begin{proof}
  We will show condition (\ref{eq:tp2c}). 
By definition of our model, 
the number of genealogical descendants in generation $t+1$ 
has a conditional binomial distribution as follows:
\[
\PP\{G_{t+1}^i = k \cond G_t^i = r\} = \binom{n}{k}
\left(2r / n - (r/n)^2 \right)^k
\left( 1 -  2r / n + (r/n)^2 \right)^{n-k}.
\]
Set $x(r) = 2r / n - (r/n)^2$, a function that
 is strictly increasing in $r$ for
$0 \leq r \leq n$. Then
\[
\frac{\PP\{G_{t+1}^i = k+1 \cond G_t^i = r\}}
	{\PP\{G_{t+1}^i = k \cond G_t^i = r\}} =
	\frac{n-k}{k+1} \cdot \frac{x(r)}{1-x(r)},
\]
a function that is a strictly increasing function of $r$ for $1 \leq r \leq n$.
\end{proof}

The following definition is well known to statisticians
\citep{lehmannTesting86}. 
\begin{defn}
  \label{defn:mlr}
  Consider a reference measure $\mu$ on some space $\mathcal{X}$
  and a parameterized family 
  $\{p_\theta : \theta \in \Theta\}$ of probability densities with
  respect to $\mu$, where $\Theta$ is a subset of $\RR$.
  Let $T$ be a real-valued function defined on $\mathcal{X}$.
  The family of densities  has the \emph{monotone
  likelihood ratio} property in $T$ with respect to the parameter
  $\theta$ if for any
  $\theta' < \theta''$ the densities $p_{\theta'}$ and 
  $p_{\theta''}$ are distinct and
  $x \mapsto p_{\theta''}(x) / p_{\theta'}(x)$ is a nondecreasing function of $T(x)$.
\end{defn}

\begin{lem}
  \label{lem:incr_expect}
  Fix a time $t \ge 0$.
  If the function $f:\RR \rightarrow \RR$ 
  is strictly increasing, then the function $k \mapsto \EE[f(G_t^i) \cond
  G_{t+1}^i = k]$, $1 \leq k \leq n$, is strictly increasing.
\end{lem}

\begin{proof}
  By Lemma~\ref{lem:tp2} and inequality (\ref{eq:tp2d}), the family of
  probability densities (with respect to counting measure)
  $\PP \left\{ G_t^i = r \cond G_{t+1}^i = k \right\}$ parameterized by 
  $k$ has monotone likelihood ratios in $r$ with respect to $k$. 
  Now apply Lemma~2(i) of \citep{lehmannTesting86}.
\end{proof}

\arxiv{
\begin{proof}[Proof of Proposition~\ref{pro:monot}]
  }
  \noarxiv{
  \noindent \emph{Proof of Proposition~\ref{pro:monot}:} \newline \noindent
  }
  For $t \ge 1$, let $\alpha_t (k) = k^{-1} \, \EE[Q_t^i \cond G_t^i = k]$.
  
  First note that each individual of $\calG_1^i$ has a $1/n$
  chance of choosing $I_{0,i}$ as a parent twice, thus 
   \[
   \alpha_1 (k) = (1 + 1/n).
  \]

  The result will thus follow by induction on $t$ if we can show 
  for $t \ge 1$ that the function
  $k \mapsto \alpha_{t+1} (k)$ is strictly increasing whenever
  the function $k \mapsto \alpha_t(k)$ is non-decreasing. 
  Therefore, fix $t \ge 1$ and suppose that the function
  $k \mapsto \alpha_t(k)$ is non-decreasing.

  We first claim that
  \begin{equation}
    k^{-1} \, \EE[Q_{t+1}^i \cond G_t^i = r, \, G_{t+1}^i = k] = 
    (r^{-1} + n^{-1}) \, \EE[Q_t^i \cond G_t^i = r] = f(r), 
    \label{eq:avpaths}
  \end{equation}
  where 
  \[
  f(r) = \left(r^{-1} + n^{-1} \right) \EE[Q_t^i \cond G_t^i = r] =
  \left( 1 + \frac{r}{n} \right) \alpha_t (r).
  \]
  The proof of this claim is as follows. 
  
 Recall that $\calG_t^i$ is the set of 
 generation $t$ individuals descended
  from $I_{0,i}$, so that $\calG_t^i$ has $G_t^i$ elements.
  Suppose that $G_t^i = r$ and number the elements of
  $\calG_t^i$ as $1,\cdots,r$. Let $V_{j,c}$ be the indicator
  random variable for the event that the allele
  $A_{t,j,c}$ is descended from one of the alleles of $I_{0,i}$ for $1
  \leq j \leq r$. By definition,
  the sum of the $V_{j,c}$ is equal to $Q_t^i$. 
  Note that any individual in $\calG_{t+1}^i$ has one parent uniformly
  selected from $\calG_t^i$ and the other uniformly selected from the
  population as a whole. Selections for different individuals are
  independent. Therefore, 
  \begin{align*}
    \EE[Q_{t+1}^i \cond G_t^i = r, \,& G_{t+1}^i = k, V_{1,1},
    V_{1,2}, \ldots, V_{r,1}, V_{r,2}] \\
    & = k \left[ 
    \frac{1}{r} \sum_{\ell=1}^r (V_{\ell,1} + V_{\ell,2}) + \frac{1}{n} \left(
    \sum_{\ell=1}^r (V_{\ell,1} + V_{\ell,2}) + \sum_{\ell =
    r+1}^n 0 \right) 
     \right] 
    \\
    & = k (r^{-1} + n^{-1}) Q_t^i
  \end{align*}
  By the tower property of conditional expectation,
  \[
  \EE[Q_{t+1}^i \cond G_t^i = r, \, G_{t+1}^i = k] = 
  k (r^{-1} + n^{-1}) \EE[Q_t^i \cond G_t^i = r, \, G_{t+1}^i = k].
  \]
  An application of the Markov property now establishes
  our claim (\ref{eq:avpaths}). Thus
  \[
  \begin{split}
    k^{-1} \, \EE[Q_{t+1}^i \cond G_{t+1}^i = k] 
    & = \sum_{r=1}^n k^{-1} \, \EE[Q_{t+1}^i \cond G_t^i = r, G_{t+1}^i = k] 
    \, \PP\{G_t^i = r \cond G_{t+1}^i = k\} \\
    & = \sum_{r=1}^n \, f(r) \, \PP\{G_t^i = r \cond G_{t+1}^i = k\} \\
    & = \EE[ f(G_t^i) \cond G_{t+1}^i = k]. \\
  \end{split}
  \]
  This is strictly increasing in $k$ by Lemma~\ref{lem:incr_expect} 
  and the observation that $f$ is strictly increasing.
\arxiv{ \end{proof} }
\noarxiv{ \QEDclosed }

\section{The mysterious shape in Figure~\ref{fig:qOverTime}}
\label{sec:mysterious}
In this section we investigate the shape of the curve relating the
number of descendant alleles to genealogical rank.  As shown in
Figure~\ref{fig:qOverTime}, this curve attains a characteristic shape
after several generations; the shape is maintained for a period
prior to the time when the genealogical MRCA appears. 
We show that this curve is essentially
the limiting ``tail-quantile'' of a normalized Poisson(2) branching
process.

An important component of our analysis will be a 
multigraph representing ancestry that we will call the {\em genealogy}.
A multigraph is similar to a graph except that multiple edges between
pairs of nodes are allowed. Specifically, a multigraph is an ordered pair
$(V,E)$ where $V$ is a set of nodes and $E$ is a multiset of unordered
pairs of nodes.

\begin{defn}
  \label{defn:ancestryGraph}
  Define the time $t$ \emph{ancestry multigraph} $\frakG_t$ as follows.  The
  nodes of this multigraph are the set of all individuals of
  generations zero through $t$; for any $0 < t' \le t$ connect an
  $I_{t',k}$ to $I_{t'-1,j}$ 
  if $I_{t',k}$ is descended from $I_{t'-1,j}$. If both parents of
  $I_{t',k}$ are $I_{t'-1,j}$, then add an additional edge connecting
  $I_{t',k}$ and $I_{t'-1,j}$. 
  Define the time $t$ \emph{genealogy} $\frakG_t^i$ to be the subgraph of
  $\frakG_t$ consisting of $I_{0,i}$ and all of its descendants
  $\bigcup_{t'=0}^t \calG_{t'}^i$ up to time $t$.
\end{defn}
  
\begin{defn}
  \label{defn:ancestryPath}
  We define an ancestry path in $\frakG_t^i$ to be a sequence of
  individuals $I_{0,i(0)}, I_{1,i(1)}, \cdots, I_{t,i(t)}$ with $i(0) = i$ where for each
  $0 < t' \leq t$, $I_{t'-1,i(t'-1)}$ is a parent of $I_{t',i(t')}$.
  Let $P_t^i$ be the number of ancestry paths in $\frakG_t^i$.
\end{defn}
We emphasize that a parent being selected twice by a single individual
results in a ``doubled'' edge;
paths that differ only in their choice of what edge to traverse
between parent to child are
considered distinct. Thus, each such doubled edge doubles the number of
ancestry paths that contain the corresponding parent-child pair.

Our result concerning the connection between the  curve in
Figure~\ref{fig:qOverTime} and the Poisson(2) branching process can be
stated as follows.
Define a random probability measure on the positive quadrant
that puts mass $1/n$ at each of the points
$(\EE[Q_t^i \cond \frakG_t], 2^{1-t} G_t^i)$.
We show below that this random probability measure converges in probability
to a deterministic probability measure concentrated
on the diagonal and has projections onto either axis given
by the limiting distribution of $2^{1-t} B_t$ as $t \rightarrow \infty$.

We may describe the convergence more concretely 
by using the idea of ``sorting by
the number of genealogical descendants'' as in the introduction;
using the notation introduced there, 
let the random variable $\sorted(t,k)$
denote the index of the individual in generation $0$ 
with the $k^{\mathrm{th}}$
greatest number of genealogical descendants
at time $t$.  Recall the non-increasing, continuous
function $\beta:(0,1) \rightarrow \RR_+$
defined in equation (\ref{E:def_beta}).

\begin{pro}
\label{pro:main}
Suppose that $0 < a < b < 1 - R(\{0\})$, so that $\infty > \beta(a) > \beta(b) > 0$.
Then
\[
\frac{1}{n(b-a)} \cdot
\#
\left\{an \le k \le bn : \EE \left[ \left. Q_t^{\sorted(t,k)} \right| \frakG_t\right] \in [\beta(b),\beta(a)]  \right\}
\]
converges to $1$ in probability as $t=t_n$ and $n$ go to infinity
in such a way that $2^{2t_n}/n \rightarrow 0$. 
\end{pro}

\noindent
Note that the condition $2^{2t_n}/n \rightarrow 0$
is satisfied, for example, when $t_n= \tau \log_2 n$ for $\tau<1/2$. 

The proof of Proposition~\ref{pro:main} formalizes the following
three common-sense notions about the ancestry
process.  

Note that for $t>1$, the genealogy will not
necessarily be a tree: it may be possible to follow two
different ancestry paths through $\frakG_t^i$ to a given time-$t$
individual. However, our first intuition is that 
this possibility is rare when $n$ is
large and $t$ is small relative to $n$, and such events 
do not affect the values of $G_t^i$ and $Q_t^i$ in the limit.

Second, the fact that each of the above genealogies is usually a tree
suggests that we may be able to relate the ancestry process to a
branching process. In our case, the
number of immediate 
descendants for an individual $I_{t'-1,j}$ is the number of
times a individual of generation $t'$ chooses $I_{t'-1,j}$ as a parent.
These numbers are not exactly independent: for example, if all of the
individuals of generation $t'$ descend only from a single individual of
generation $t'-1$, then the number of descendants of the other
individuals is exactly zero. However, we will show that these numbers
 are close
to independent when $n$ becomes large. Also, note that the marginal
distribution of the number of next-generation descendants of a single
individual is binomial: there are $2n$ trials each with probability
$1/n$. As $n$ goes to infinity, this is approximately a Poisson(2) random
variable. In summary, we will show that the genealogy of an
individual is close to that of a Poisson(2) branching process for
short times relative to the population size. 

Third, we note that there is a simple relationship between the number
of paths $P_t^i$ and the expected number of descendant alleles $Q_t^i$:
\begin{lem}
  \label{lem:pAndQ}
  $\EE\left[\Qit \cond \frakG_t^i\right] = 2^{1-t} \Pit$.
\end{lem}

\begin{proof}
  Consider an arbitrary path in the ancestry graph $\frakG_t^i$ and pick
  an arbitrary edge in that path. 
  Suppose the edge connects $I_{t'-1,j}$ to $I_{t',i}$. 
  By the definition of the model, $I_{t',i}$ 
  has probability $1/2$  of inheriting any fixed allele of
  $I_{t'-1,j}$. Thus, the 
  contribution of any single allele of $I_{0,i}$ and 
  given path in $\frakG_t^i$ to the expectation of 
  $Q_t^{i}$ is $2^{-t}$. The contribution of both alleles of
  $I_{0,i}$ is $2^{1-t}$.
  The total number of alleles descended from the alleles of $I_{0,i}$
  is the sum over the contributions of all paths, and the expectation
  of this sum is the sum of the expectations. 
\end{proof}

We will use the 
probabilistic method of coupling to formalize the connection between
the genealogical process and the branching process.  A coupling of random variables $X$
and $Y$ that are not necessarily defined on the same probability
space is a pair of random variables $X'$ and $Y'$ defined on a single
probability
space such that the marginal distributions of $X$ and $X'$
(respectively, $Y'$ and $Y$) are the same.  A simple example of
coupling is ``Poisson thinning'', a coupling between an $X
\sim \poi{\lambda_1}$ and a $Y \sim \poi{\lambda_2}$ where $\lambda_1
\geq \lambda_2$. To construct the pair $(X',Y')$, one first
gains a sample for $X'$ by simply sampling from $X$. The sample from
$Y'$ is then gained by ``throwing away'' points from the sample for
$X'$ with probability $\lambda_2 / \lambda_1$; i.e. the distribution
for $Y'$ conditioned on the value $x$ for
$X'$ is just $\Bin(x, 1 - \lambda_2 / \lambda_1).$ 

We note that coupling is a popular tool for questions with
a flavour similar to ours.
Recently \citet{barbour07} has coupled an epidemics model to a branching
process  and \citet{durrettEa07} have used coupling to
analyze a model of carcinogenesis.

Recall that we defined $W_t = B_t / 2^t$, where $B_t$ is
a Poisson(2) branching processes started at time $t=0$
from a single individual, 
and we observed that the sequence of random variables
$W_t$ converges almost surely to a random variable $W$ with
distribution $R$. The
following lemma is the coupling result that will give the convergence of
the sampling distribution of the $\Pit$ and $\Git$ to $R$ in
Lemma~\ref{lem:etaConv} below.

\begin{lem}
  \label{lem:joint}
There is a coupling between $P_t^i$, $G_t^i$, and $B_t^i$, where
$B_t^1, B_t^2, \ldots$ is a sequence of independent
Poisson(2) branching processes, such that
for a fixed positive integer $\ell$ the probability
\[
\PP\{P_t^i = G_t^i = B_t^i, \, 1 \le i \le \ell\}
\]
converges to one as $n$ goes to infinity with
$t=t_n$ satisfying $2^{2t_n}/n \rightarrow 0$. 
\end{lem}

\begin{proof}
We introduce the coupling between the ancestral process and the branching process
by looking first at 
the transition from generation $0$ to generation $1$.  Suppose that we
designate a set $S$ of $k$ individuals in generation $0$ and write $G$ for the
number of descendants these $k$ individuals have in generation $1$.

The probability that there is an individual in generation $1$ who
picks both of its parents from the $k$ designated individuals is
\[
1 - \left(1 - \left (k/n \right)^2 \right)^n \le k^2/n.
\]

Couple the random variable $G$ with a random variable $P$ 
that is the same as $G$ except that we (potentially repeatedly)
re-sample any generation $1$ individual who
chooses two parents from $S$ until it has at least one parent not
belonging to $S$. 
The random variable $P$ will have a binomial distribution 
with number of trials $n$ and success probability
\begin{equation}
\frac{2 \frac{k}{n} \left(1 - \frac{k}{n}\right)}{1-\frac{k^2}{n^2}}
=
\frac{2 \frac{k}{n}}{1 + \frac{k}{n}},
\end{equation}
which is simply the probability of an individual selecting exactly one
parent from the set of $k$ given that it does not select two.
By the above,
\[
\PP\{G \ne P\} \le \frac{k^2}{n}.
\]

By a special case of Le Cam's Poisson approximation result
\citep{MR2059709, MR0142174}, we can couple the random variable
$P$ to a random variable $Y$ that is Poisson distributed with
mean
\[
n \frac{2 \frac{k}{n}}{1 + \frac{k}{n}} = \frac{2 k}{1 + \frac{k}{n}}
\]
in such a way that
\[
\PP\{P \ne Y\} \le n \left(\frac{2 \frac{k}{n}}{1 + \frac{k}{n}}\right)^2
\le 4 \frac{k^2}{n}.
\]

Moreover, a straightforward argument using Poisson thinning shows that
we can couple the random variable $Y$ with a random variable $B$ that is Poisson distributed with
mean $2k$ such that
\[
\PP\{Y \ne B\} \le \left|2k - \frac{2 k}{1 + \frac{k}{n}}\right| \le 2 \frac{k^2}{n}.
\]

Putting this all together, we see that we can couple the random 
variables $G$, $P$, and $B$ together
in such a way that
\[
\PP(\neg \{G = P = B\}) \le 8 \frac{k^2}{n}
\]
where $\neg$ denotes complement.
Note that $B$ may be thought of as the sum of $k$ independent
random variables, each having a Poisson distribution with mean $2$.

Fix an index $i$ with $1 \leq i \leq n$.
Returning to the notation used in the rest of the paper, the above
triple $(G,P,B)$ correspond to $(G_t^i,P_t^i,B_t^i)$, and $k$ plays the
role of $G_{t-1}^i$.
Now suppose we start with one designated individual $i$ in the population
at generation $0$. 
Let $S_t$ denote the event
\[
\{P_t^i = G_t^i = B_t^i\}.
\]
The above argument shows that we can couple the
process $P^i$ with the branching process $B^i$ in such a way that
\[
\begin{split}
  \PP\{\neg S_t \}
& \le 
\PP\{\neg S_{t-1}\} 
+ 
\PP\{\neg S_t, \; S_{t-1}\} \\
& \le 
\PP\{\neg S_{t-1}\} 
+ \EE\left[8 \frac{B_{t-1}^2}{n} \right] \\
& \le
\EE\{\neg S_{t-1}\} + \frac{c 2^{2(t-1)}}{n} \\
\end{split}
\]
for a suitable constant $c$ (using standard formulae for moments
of branching processes).  Iterating this bound gives
\[
\PP\{\neg S_t\} \le \frac{c' 2^{2t}}{n}
\]
for a suitable constant $c'$.

This tells us that when $n$ is large, the random
variable $P_t^i$ is close to the random variable
$B_t^i$ not just for fixed times but more generally for times $t$ such that
$2^{2t} / n \rightarrow 0$. As mentioned above, this condition is
satisfied when $t= \tau \log_2 n$ for $\tau<1/2$. 

Next, we elaborate the above argument to handle
the descendants of $\ell$ individuals. Let $S_t^\ell$ denote the
event that the $\ell$ coupled triples of random variables are equal, that is,
\[
\{P_t^1 = G_t^1 = B_t^1, \, P_t^2 = G_t^2 = B_t^2, \ldots, P_t^\ell =
G_t^\ell = B_t^\ell\},
\]
where $B_t^i$ is the branching process coupled to $P_t^i$ and $G_t^i$.
By mimicking the above argument, we can show that
\[
  \PP\{\neg S_t^\ell\} \leq \PP\{\neg S_{t-1}^\ell\} + c \ell \,
  2^{2(t-1)} / n.
\]
Again, iterating this bound gets
\[
  \PP\{\neg S_t^\ell\} \leq c' \ell \, 2^{2t} / n
\]
for some $c'$. 
\end{proof}

For any Borel subset $C$ of $\RR_+^2$,
let $\eta_{t,n}(C)$ denote the joint empirical distribution of
the normalized $P_t^i$ and the normalized $G_t^i$ at time $t$, i.e.
\[
\eta_{t,n}(C) = \frac{1}{n} \cdot \#\{1 \le i \le n: (2^{-t} P_t^i, 2^{-t} G_t^i) \in C\}.
\]
In Lemma~\ref{lem:etaConv} we demonstrate that the $\eta_{t,n}$ converge in
probability to the deterministic probability measure
$\eta(dx,dy) = R(dx) \delta_x(dy) = \delta_y(dx) R(dy)$ 
concentrated on the
diagonal, where $\delta_z$ denotes the unit point mass at $z$.

The mode of convergence may require a bit of explanation. When we say that
a real-valued random variable converges in probability to a fixed
quantity, there is an implicit and commonly understood notion of
convergence of a sequence of real numbers. However, here the random
quantities are probability measures, and the underlying
notion we use for
convergence of measures is that of weak convergence. 
Recall that a sequence
of probability measures $\mu_n$ on $\RR_+^2$
is said to converge to $\mu$ weakly if $\int f \, d\mu_n$ converges to
$\int f \, d\mu$ for all bounded continuous functions $f: \RR_+^2 \rightarrow \RR$.
The following are equivalent conditions for sequence of probability measures
$\mu_n$ to converge to $\mu$ weakly: (i) $\limsup_n \mu_n(F) \le \mu(F)$
for all closed sets $F \subseteq \RR_+^2$, (ii) $\liminf_n \mu_n(G) \ge \mu(G)$ for all  open
sets $G \subseteq \RR_+^2$, (iii) $\lim_n \mu_n(A) = \mu(A)$ for all Borel sets $A \subseteq \RR_+^2$
such that $\mu(\partial A) = 0$, where $\partial A$ is the boundary of $A$.

\begin{lem}
  \label{lem:etaConv}
  Suppose that $t=t_n$ converges to infinity
  as n goes to infinity in such a way that
  $\lim_{n \rightarrow \infty} 2^{2t_n}/n \rightarrow 0$.
  Then the sequence of random measures 
  $\eta_{t,n}$ converges in
  probability as $n \rightarrow \infty$
  to the deterministic probability measure $\eta$
  on $\RR_+^2$
  that assigns mass $R(A \cap B)$ to sets of the form $A \times B$.
\end{lem}

\begin{proof}
For brevity, let $H_t^i$ denote the pair $(2^{-t} P_t^i, 2^{-t} G_t^i)$.
Fix a bounded continuous function $f: \RR_+^2 \rightarrow \RR$.
By definition, 
\[
\begin{split}
\EE\left[\left(\int f \, d \eta_{t,n}\right)^2\right] \\
& = n^{-2} \, \EE \left[\sum_i f^2(H_t^i) 
+ \sum_{i \ne j} f(H_t^i) f(H_t^j) \right] \\
& = n^{-2} \left(n \EE[f^2(H_t^1)] + n(n-1) \EE[f(H_t^1) f(H_t^2)] \right). \\
\end{split}
\]
Hence,
$\EE[(\int f \, d \eta_{t,n})^2]$
is asymptotically equivalent to
\begin{equation}
  \EE[f(H_t^1) f(H_t^2)].
  \label{eq:Ajoint}
\end{equation}
By definition, (\ref{eq:Ajoint}) is equal to
\begin{equation}
  \EE[f(2^{-t} P_t^1, 2^{-t} G_t^1) \, f(2^{-t} P_t^2, 2^{-t} G_t^2)].
  \label{eq:PandG}
\end{equation}
Lemma~\ref{lem:joint} establishes a coupling such that $P_t^i = G_t^i
= B_t^i$ with probability tending to one in the limit
under our hypotheses. Thus, under our conditions on $t = t_n$
the expectation (\ref{eq:PandG}), and hence 
$\EE[(\int f \, d \eta_{t,n})^2]$, converges to
\[
\begin{split}
\lim_{t \rightarrow \infty} \EE[f(W_t^1) f(W_t^2)] 
& = 
\lim_{t \rightarrow \infty} \EE[f(W_t^1, W_t^1)] \EE[f(W_t^2, W_t^2)] \\
& = \left(\int_{\RR_+} f(x,x) \, R(dx)\right)^2 \\
& = \left(\int_{\RR_+^2} f \, d \eta\right)^2. \\
\end{split}
\]

A similar but simpler argument shows that
$\EE[\int f \, d \eta_{t,n}]$ converges to $\int f \, d \eta$. Combining
these two facts shows that
$\var[\int f \, d \eta_{t,n}]$ converges to zero.

Therefore, $\int f \, d \eta_{t,n}$ converges in probability to
$\int f \, d \eta$ for all bounded continuous functions $f$, as required.
\end{proof}

\arxiv{ \begin{proof}[Proof of Proposition~\ref{pro:main}] }
  \noarxiv{ \emph{Proof of Proposition~\ref{pro:main}:} \noindent }

  It suffices by Lemma~\ref{lem:pAndQ} to show that
  \[
\frac{1}{n(b-a)}
\cdot
\#
\left\{an \le k \le bn : 2^{-t} P_t^{\sorted(t,k)}  \in [2\beta(b),2\beta(a)]  \right\}
\]
converges to $1$ in probability  as $t=t_n$ and $n$ go to infinity
in such a way that $2^{2t_n}/n \rightarrow 0$.

For $\gamma > 0$ and an integer $1 \le k \le n$,
\[
\begin{split}
& \eta_{t,n}(\RR_+ \times [\gamma,\infty)) \ge \frac{k}{n} \\
& \quad \Leftrightarrow
\#\{1 \le i \le n : 2^{-t} G_t^i \ge \gamma\} \ge k \\
& \quad \Leftrightarrow 
2^{-t} G_t^{\sorted(t,k)} \ge \gamma, \\
\end{split}
\]
by definition of the empirical distribution
$\eta_{t,n}$ and the indices $\sorted(t,k)$.
Because the limit measure $\eta$
assigns zero mass to the boundary $\RR_+ \times \{\gamma\}$
of the set $\RR_+ \times [\gamma,\infty)$, it follows
from Lemma~\ref{lem:etaConv} that 
\[
\frac{1}{n} \cdot \#\{1 \le i \le n : 2^{-t} G_t^i \ge \gamma\}
\]
converges to 
$\eta(\RR_+ \times [\gamma,\infty)) = R([\gamma,\infty))$
in probability.
In particular,
\[
\frac{1}{n} \cdot \#\{1 \le i \le n : 2^{-t} G_t^i \ge 2 \beta(c)\}
\]
converges to $c$ in probability
for $0 < c < 1 - R(\{0\})$.  Thus,
$2^{-t} G_t^{\sorted(t, \lfloor cn \rfloor)}$ converges
in probability to $2 \beta(c)$ for such a $c$.

With $0 < a < b < 1 - R(\{0\})$ as in the statement
of the proposition, it follows that
\[
\frac{1}{n}
\cdot \#
\left\{an \le k \le bn : 2^{-t} G_t^{\sorted(t,k)}  \in [2\beta(b-\epsilon),2\beta(a+\epsilon)]  \right\}
\]
converges in probability to $(b - a - 2 \epsilon)$
for $0 < \epsilon < (b-a)/2$.

Note by Lemma~\ref{lem:etaConv} that 
\[
\begin{split}
& \frac{1}{n}
\cdot \#
\left\{1 \le k \le n : 
\left| 2^{-t} P_t^{\sorted(t,k)} - 2^{-t} G_t^{\sorted(t,k)}\right|
> \delta  \right\} \\
& \quad =
\eta_{t,n}(\{(x,y) \in \RR_+^2 : |x-y| > \delta\}) \\
\end{split}
\]
converges in probability to $0$ for any $\delta>0$, because
the probability measure $\eta$ assigns all of its mass
to the diagonal $\{(x,y) \in \RR_+^2 : x=y \}$.  

Taking
$\delta < 2 \min\{\beta(a) - \beta(a+\epsilon), \beta(b-\epsilon) - \beta(b)\}$
so that
\[
[ 2 \beta (b-\varepsilon) - \delta, 2 \beta (a+\varepsilon) + \delta ]
\subseteq [ 2 \beta (b), 2 \beta (a) ],
\]
letting $n$ tend to infinity, and then sending
$\epsilon$ to zero completes the proof.

  \arxiv{\end{proof}}
  \noarxiv{\QEDclosed}

As an application of this proposition, one might wonder about the
number of descendant alleles of those individuals with many
genealogical descendants. It is imaginable that the number of
descendant alleles of each individual would stay bounded;
however, this is not the case. 

\begin{cor}
  Fix $y>0$, and suppose $t=t_n$ satisfies $\lim_{n
  \rightarrow \infty} 2^{2t_n}/n \rightarrow 0$.  With probability
  tending to one as $n$ goes to infinity, there will be an
  individual $i$ in the population at time $0$  
  such that $\EE \left[ \left. Q_t^i \right| \frakG_t\right] >
  y$.
\end{cor}

\begin{proof}
  Because the support of the probability distribution $R$
  is all of $\RR_+$, the function $\beta$ is unbounded.  The result
   is then immediate from Proposition~\ref{pro:main} . 
\end{proof}

\section{The number of MRCAs and the number of descendant alleles per MRCA}
\label{sec:mrca}

There are a number of other interesting phenomena that seem more
difficult to investigate analytically but are interesting enough to deserve mention.
For the simulations of this section (and the one mentioned in the
introduction) we wrote a series of simple
\texttt{ocaml} programs which are available upon request. 

As mentioned in the introduction, it is not uncommon to get several
genealogical MRCAs simultaneously.  We denote the (random) time to 
achieve a genealogical MRCA for a
population of size $n$ by $\tmrca_n$. We denote the (random) number
of genealogical MRCAs for a population of size $n$ by $\nmrca_n$. The 
surprising dependence
of $\EE[\nmrca_n]$ on $n$ is shown in
Figure~\ref{fig:totNMrca}.

\begin{figure}
  \begin{center}
    \arxiv{
    \includegraphics[height=5cm]{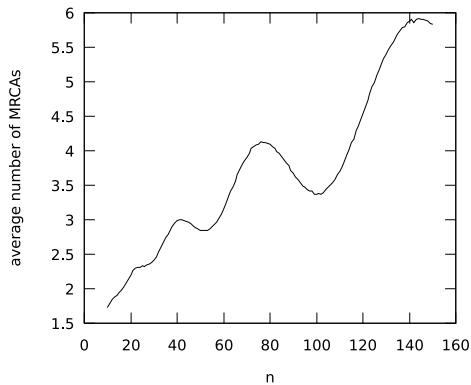}
    }
  \end{center}
  \caption{The dependence of the expected number of MRCAs
  on population size. Average of 10000 simulations.}
  \label{fig:totNMrca}
\end{figure}

However, the situation becomes clear by investigating the conditional
expectation $\EE[M_n | T_n]$ as shown in
Figure~\ref{fig:nMrcaCond}. According to the law of total expectation,
one can gain the expectation by taking the sum of conditional
expectations weighted by their probability. In this setting,
\begin{equation}
  \EE[\nmrca_n] = \sum_{k} \EE[\nmrca_n \cond \tmrca_n = k] \,
  \PP\{\tmrca_n = k\}
  \label{eq:mrcaTotalProb}
\end{equation}
First note in Figure~\ref{fig:nMrcaCond}~(a) that $\EE[\nmrca_n \cond
\tmrca_n = k]$ appears to be a
decreasing function of $n$ when $k$ is fixed. This is not too
surprising: imagine that we are doing simulations with $n$ individuals, but only looking at the
results of simulations such that $\tmrca_n = k$. When $n$ gets large,
simulations such that $\tmrca_n = k$ are ones which take an unusually
short time to reach $\tmrca$. It's not surprising to find that the
number of MRCAs would be small in this
case.
Conversely, simulations such that $\tmrca_n$ is significantly bigger
than $\log_2 n$ are ones that take an unusually long time; it is not surprising that
such simulations have a larger number of MRCAs as they have more
individuals ``ready'' to become MRCAs just before $\tmrca_n$.
This argument is bolstered by Figure~\ref{fig:ancestryAccumulation}
which shows that simulations resulting in different $\tmrca_n$'s have
remarkably similar behavior. Specifically, the distribution of the
number of genealogical descendants sorted by rank does not show a very
strong dependence on the time to most recent common ancestor
$\tmrca_n$.
Therefore simulations for a given population size that have a smaller $\tmrca_n$ have fewer
individuals who are close to being MRCAs while individuals with
larger $\tmrca_n$ have more.

\begin{figure}
  \begin{center}
    \arxiv{
    \includegraphics[height=5cm]{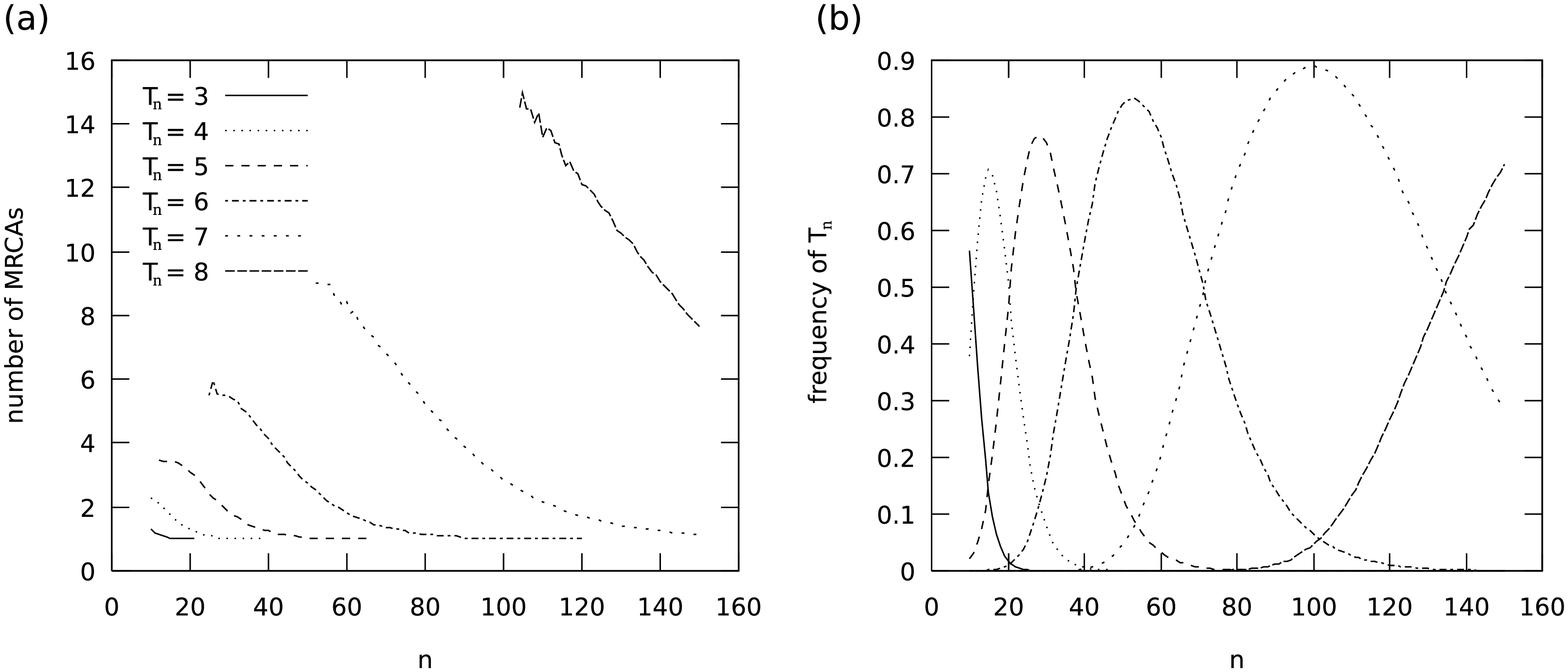}
    }
  \end{center}
  \caption{The number of MRCAs where the dependence on $\tmrca_n$ (the
  time to MRCA) is taken into account. Average of 10000 simulations.
  Figure (a) shows the number of MRCAs at $\tmrca_n$ conditioned on
  $\tmrca_n$.  Figure (b) shows the dependence of the distribution of
  times to MRCA on population size.  As described in the text, it is
  the combination of these two distributions using the law of total
  expectation that produces the ``bumps'' of
  Figure~\ref{fig:totNMrca}. Note that several simulations with
  ``extreme'' values of $\tmrca$ have been eliminated from (a) for
  clarity; these combinations of $\tmrca_n$ and population size are rare
  and thus we would not get an accurate estimate of the expectation.
  }
  \label{fig:nMrcaCond}
\end{figure}

Second, note in Figure~\ref{fig:nMrcaCond}~(b) that the distribution of $T_n$ has bumps such that 
(at least for integers $k>3$), there is an interval of $n$ such that $\PP\{T_n
= k\}$ is large in that interval. In such an interval we are
approximately on a single line of Figure~\ref{fig:nMrcaCond} (a), that
is, $\EE[\nmrca_n]$ is approximately $\EE[\nmrca_n \cond \tmrca_n =
\kappa_n]$ where $\kappa_n$ is the most likely value of $\tmrca_n$. 
This value is decreasing as described in the previous paragraph; thus
we should see a dip in $\EE[\nmrca_n]$. Indeed, from the plots of
Figures~\ref{fig:totNMrca} and \ref{fig:nMrcaCond} (b) it can be seen
that the dips in the number of MRCAs correspond to the peaks
of the probability of a given $\tmrca_n$. 

\begin{figure}
  \begin{center}
    \arxiv{
    \includegraphics[height=5cm]{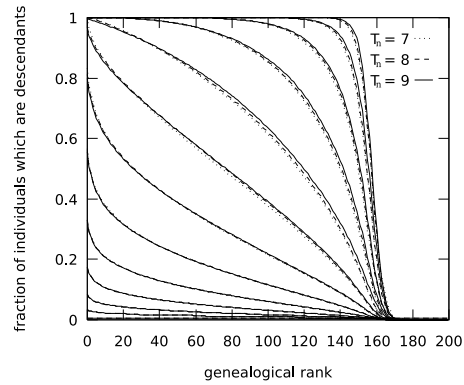}
    }
  \end{center}
  \caption{A plot of $\EE \left[ \left. G_t^{\sorted(t,k)} \, \right| \, \tmrca_n \right]$ through time conditioning on $\tmrca_n$.
  Each curve for a given choice
  of $T_n$ represents the expected state of the process at a given
  time. That is, each curve represents
  the image of the map $k \mapsto 
  \EE \left[ \left. G_t^{\sorted(t,k)} \, \right| \, \tmrca_n \right]$
  for some choice of $t$ and $\tmrca_n$.
  As described in the text, the curves show surprisingly little
  dependence on $\tmrca_n$, rather depending almost exclusively on
  $t$. Average of 10000 simulations with $n=200$.}
  \label{fig:ancestryAccumulation}
\end{figure}

Now we return to the genetic story considered in the rest of the
paper. The above considerations certainly apply when formalizing
questions such as ``how genetically related is the MRCA to individuals
of the present day?'' Clearly, there will often not be only one MRCA
but a number of them. Furthermore, the dynamics of the numbers of
MRCAs plays an important part in the answer to the question. 

In Figure~\ref{fig:totNTok} we show the number of alleles descended
from 
the union of the MRCAs as a function of $n$. This shows
oscillatory behavior as in Figure~\ref{fig:totNMrca}, however the
effect is modulated by the results shown in Figure~\ref{fig:tokAvg}.
Specifically, although the number of MRCAs decreases with $n$
conditioned on a value of $\tmrca_n$, the number of descendant alleles per MRCA is
actually increasing. The combination of these two functions appears to 
still be a decreasing function, which creates the ``dips'' in
Figure~\ref{fig:totNTok}. 

\begin{figure}
  \begin{center}
    \arxiv{
    \includegraphics[height=5cm]{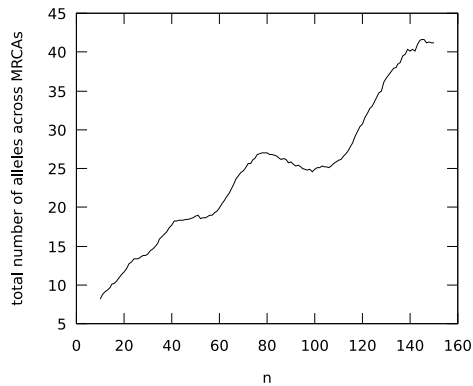}
    }
  \end{center}
  \caption{The number of alleles descended from the union of the MRCAs
  versus population size. This plot shows oscillatory behavior
  similar to that in Figure~\ref{fig:totNMrca} but the effect is dampened by the
  fact that the average number of alleles descended from each MRCA increases
  with $n$ as shown in Figure~\ref{fig:tokAvg}. Average of 10000
  simulations. Some simulations with ``extreme'' values of $\tmrca$
  were excluded for clarity as in Figure~\ref{fig:nMrcaCond}~(a). }
  \label{fig:totNTok}
\end{figure}

The apparent fact that, while fixing $\tmrca_n$, the average number of alleles descended from each MRCA 
appears to increase with $n$ deserves some explanation. As demonstrated in
Lemma~\ref{lem:pAndQ}, the expected number of descendant alleles of an individual
is a multiple of the number of paths to present-day ancestors in the
genealogy that individual.
Therefore, the fact needing explanation is the apparent increase in the
number of paths as $n$ increases. This can be explained in a way
similar to that for the conditioned number of MRCAs. Let us again fix 
$\tmrca_n = \kappa$ and vary $n$. When $n$ gets large, simulations
which the required value of $\tmrca_n$ have found a common ancestor
quite quickly. In these cases the ``ancient'' endpoints of the paths should be
tightly focused in
the most recent common ancestors. On the other hand, for small $n$ the
simulations have reached $\tmrca_n$ relatively slowly so the
distribution of paths is more diffuse. 
\begin{figure}
  \begin{center}
    \arxiv{
    \includegraphics[height=5cm]{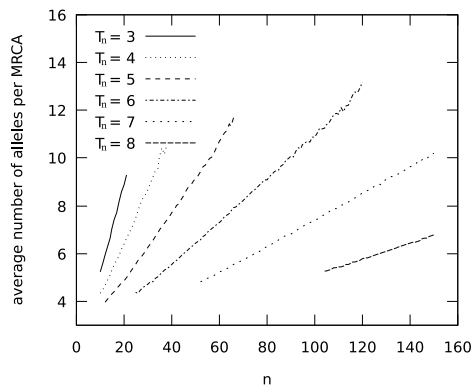}
    }
  \end{center}
  \caption{The average number of alleles descended from each MRCA conditioned on $\tmrca$.}
  \label{fig:tokAvg}
\end{figure}

\section{Conclusion}

We have investigated the connection between genetic ancestry and
genealogical ancestry in a natural genetic model extending the genealogical
model of \citet{chang99}. We have shown that an increased number
of genealogical descendants implies a super-linear increase in the number of
descendant alleles. We have tracked how the number of genetic
descendants depends on the number of genealogical descendants through time
and shown that it acquires an understandable shape for a period of
time before $\tmrca_n$ (the time of the genealogical MRCA). 
We have also investigated the
number of MRCAs at $\tmrca_n$, and the number of alleles descending
from the MRCAs, and explained their surprising oscillatory
dependence on population size using simulations. 

\subsubsection*{Acknowledgements}
The authors would like to thank C. Randal Linder and Tandy Warnow for
suggesting the idea of investigating the connection between genealogy
and genetics in this setting.
John Wakeley was a collaborator for some early explorations of these
questions and we gratefully acknowledge his part in this work.
Montgomery Slatkin has provided encouragement and interesting questions.
\noarxiv{
FAM is funded by the Miller Institute for Basic Research in
Science at the University of California, Berkeley.  SNE is supported
in part by NSF grant DMS-0405778. Part of this research was conducted
during a visit to the Pacific Institute for the Mathematical
Sciences.
}

\noarxiv{ \bibliographystyle{tpb} }
\arxiv{ \bibliographystyle{plainnat} }
\bibliography{mrca}

\end{document}